\documentclass[11pt]{article} 
\usepackage{times}
\usepackage{amsmath,amsfonts}
\usepackage{epsfig,amssymb,amstext,xspace}
\usepackage{amsthm}
\usepackage[noend]{algorithmic}
\usepackage{algorithm}

\newcommand{\E}[2][{}]{\ensuremath{\mathrm{E}_{#1}\bigl[#2\bigr]}}   
\newcommand{\Pro}[1]{\ensuremath{\Pr\bigl[#1\bigr]}}

\newcommand{\np}{{\em NP}\xspace}
\newcommand{\nphard}{\np-hard\xspace} 

\newcommand{\apx}{{\em APX}\xspace}

\newtheorem{theorem}{Theorem}[section]
\newtheorem{fact}[theorem]{Fact}
\newtheorem{lemma}[theorem]{Lemma}
\newtheorem{corollary}[theorem]{Corollary}
\newtheorem{myclaim}[theorem]{Claim}{\bfseries}{\itshape}
\newtheorem{proposition}[theorem]{Proposition}

\newtheorem{definition}[theorem]{Definition}

{\theoremstyle{remark} \newtheorem{remark}[theorem]{Remark}}
{\theoremstyle{definition}

}

\newcommand{\bg}[1]{\medskip\noindent{\bf #1}}

\newenvironment{proofof}[1]{\begin{proof}[Proof of #1]}{\end{proof}}

\def\EE{\mathbb E}

\newcommand{\hide}[1]{}

\newcommand{\cP}{\ensuremath{\mathcal{P}}}

\newcommand{\cR}{\ensuremath{\mathcal{R}}}
\newcommand{\cE}{\ensuremath{\mathcal{E}} }

\newcommand{\R}{\ensuremath{\mathbb R}}
\newcommand{\Z}{\ensuremath{\mathbb Z}}

\newcommand{\A}{\ensuremath{\mathcal{A}}}

\newcommand{\Hc}{\ensuremath{\mathcal{H}}}
\newcommand{\I}{\ensuremath{\mathcal I}}

\newcommand{\Sc}{\ensuremath{\mathcal S}}
\newcommand{\Pc}{\ensuremath{\mathcal P}}
\newcommand{\Rc}{\ensuremath{\mathcal R}}
\newcommand{\V}{\ensuremath{\mathcal V}}

\newcommand{\OPT}{\ensuremath{\textsc{Opt}}}

\newcommand{\frall}{\ensuremath{\text{ for all }}}

\newcommand{\sse}{\ensuremath{\subseteq}}
\newcommand{\sm}{\ensuremath{\setminus}}
\newcommand{\es}{\ensuremath{\emptyset}}

\newcommand{\poly}{\operatorname{poly}}

\newcommand{\e}{\ensuremath{\epsilon}}
\def\eps{\varepsilon}
\newcommand{\gm}{\ensuremath{\gamma}}

\newcommand{\al}{\ensuremath{\alpha}}
\newcommand{\ld}{\ensuremath{\lambda}}
\newcommand{\tht}{\ensuremath{\theta}}
\newcommand{\Tht}{\ensuremath{\Theta}}
\def\br#1{{{(#1)}}}

\newcommand{\hx}{\ensuremath{\hat x}}

\newcommand{\bone}{\ensuremath{\boldsymbol{1}}}

\newcommand{\round}{\ensuremath{\mathsf{Round}}}
\newcommand{\opt}{\ensuremath{\mathsf{opt}}}
\newcommand{\swm}{\ensuremath{\mathrm{SWM}}}

\newcommand{\plp}[1]{\ensuremath{(\text{\ref{ca-p}}_{#1})}}
\newcommand{\dlp}[1]{\ensuremath{(\text{\ref{ca-d}}_{#1})}}

\title{Approximation Algorithms for Non-Single-minded Profit-Maximization
  Problems with Limited Supply\footnote{A short version of this paper appeared in the Proceedings of the 6th International Workshop on Internet and Network Economics (WINE 2010).}} 
\author{Khaled Elbassioni \thanks{Masdar Institute of Science and Technology, P.O.Box 54224, Abu Dhabi, UAE;
(kelbassioni@masdar.ac.ae).} 
\and 
Mahmoud Fouz \thanks{FR Informatik, Universit\"at des Saarlandes, D-66123, Saarbr\"ucken, Germany; (fmahmoud@mpi-inf.mpg.de).}
\and
Chaitanya Swamy
\thanks{Dept. of Combinatorics and Optimization, Univ. Waterloo, Waterloo, ON N2L 3G1; (cswamy@math.uwaterloo.ca).
Supported by NSERC grant 327620-09 and an Ontario Early Researcher Award.}  
}

\date{}

\begin{document}

\maketitle
\def\thepage{}
\thispagestyle{empty}

\begin{abstract}
We consider {\em profit-maximization} problems for {\em combinatorial auctions} with 
{\em non-single minded valuation functions} and {\em limited supply}. 
There are $n$ customers and $m$ items, each of which is available is in some limited
supply or capacity. Each customer $j$ has a value $v_j(S)$ for each subset $S$ of
items specifying the maximum amount she is willing to pay for that set (with
$v_j(\es)=0$). A feasible solution to the profit-maximization problem consists of item
prices 
and an allocation $(S_1,\ldots,S_n)$ of items to customers such that (i) the price of the
set $S_j$ assigned to $j$ is at most $v_j(S_j)$, and (ii) the number of customers who are
allotted an item is at most its capacity. The goal is find a feasible
solution that maximizes the total
profit earned by selling items to customers.

We obtain fairly general results that relate the approximability of the
profit-maximization problem to that of the corresponding {\em social-welfare-maximization}
(SWM) problem, which is the problem  of finding an allocation $(S_1,\ldots,S_n)$
satisfying the capacity constraints that has maximum total value $\sum_j v_j(S_j)$. 
For {\em subadditive valuations} 
(and hence {\em submodular, XOS valuations}), we obtain a solution with profit 
$\OPT_\swm/O(\log c_{\max})$, where $\OPT_\swm$ is the optimum social welfare and
$c_{\max}$ is the maximum item-supply; thus, this yields an $O(\log
c_{\max})$-approximation for the profit-maximization problem. 
Furthermore, given {\em any} class of valuation functions, if
the SWM problem for this valuation class has an LP-relaxation (of a certain form) and an
algorithm ``verifying'' an {\em integrality gap} of $\al$ for this LP, then we 
obtain a solution with profit $\OPT_\swm/O(\al\log c_{\max})$, thus obtaining an
$O(\al\log c_{\max})$-approximation. 

The latter result immediately yields an $O(\sqrt m\log c_{\max})$-approximation for
the profit maximization problem for combinatorial auctions with 
{\em arbitrary valuations}. As another application of this result, we consider the
non-single-minded tollbooth problem on trees (where items are edges of a tree, and
customers desire paths of the tree). We devise an $O(1)$-approximation algorithm for the
corresponding SWM problem satisfying the desired integrality-gap requirement, and thereby
obtain an $O(\log c_{\max})$-approximation for the non-single-minded tollbooth problem on
trees. For the special case, when the tree is a path, we also obtain an incomparable
$O(\log m)$-approximation (via a different approach) for subadditive valuations, and
arbitrary valuations with unlimited supply. Our approach for the latter problem also gives an $\frac{e}{e-1}$-approximation algorithm for the multi-product
  pricing problem in the Max-Buy model, with limited supply, improving on the previously known approximation factor of 2.
\end{abstract}

\section{Introduction}\label{Intro}
Profit (or revenue) maximization is a classic and fundamental economic goal, and the
design of computationally-efficient item-pricing schemes for 
various profit-maximization problems has received much recent attention%
~\cite{AFMZ04,GHKKKM05,BB07,AH06,BK06,BBM08}.     
We study the algorithmic problem of {\em item-pricing for profit-maximization} for 
{\em general} (multi unit) {\em combinatorial auctions} (CAs) with {\em limited supply}.  
There are $n$ customers and $m$ items. Each item is available is in some limited supply or
capacity, and each customer $j$ has a value $v_j(S)$ for each subset $S$ of
items specifying the maximum amount she is willing to pay for that set (with
$v_j(\es)=0$).  
Given a pricing of the items, a \emph{feasible allocation} is an assignment of a (possibly 
empty) subset $S_j$ to each customer $j$ satisfying (i) the {\em budget constraints},
which require that the price of $S_j$ (i.e., the total price of the items in $S_j$) is at
most $v_j(S_j)$, and (ii) the {\em capacity constraints}, which stipulate that the number
of customers who are allocated an item be at most the supply of that item.
The objective is to determine item prices that maximize the total profit or revenue earned
by selling items to customers.
Guruswami et al.~\cite{GHKKKM05} introduced the {\em envy-free}
version of the problem, where there is the additional constraint that the set assigned to 
a customer must maximize her utility (defined as value$-$price). 
Item pricing has an appealing simplicity and enforces a basic notion of fairness wherein
the seller does not discriminate between customers who get the same item(s). 
Our focus on item pricing is in keeping with the vast majority of 
work on algorithms for profit-maximization (e.g., the above references; in fact, with
unlimited supply and unit-demand valuations, our problem reduces to the {\em Max-Buy}
model in~\cite{AFMZ04}). 
Various current trading practices are described by item pricing, and thus it
becomes pertinent to understand what guarantees are obtainable via such schemes.  
Profit-maximization problems are typically \nphard, 
so we will be interested in designing approximation algorithms for these
problems. Throughout, a $\rho$-approximation algorithm for a maximization
problem, where $\rho\geq 1$, denotes a polytime algorithm that returns a solution
of value at least $(\text{optimum value})/\rho$.  

The framework of combinatorial auctions is an extremely rich framework that encapsulates a 
variety of applications. In fact, recognizing the generality of the envy-free
profit-maximization problem for CAs, Guruswami et al. proceeded to study various 
more-tractable special cases of the problem. In particular, they introduced the following 
two structured problems in the {\em single-minded} (SM) setting, where each customer desires a
single fixed set: 
(a) the {\em tollbooth problem} where the items are edges of a graph and the
customer-sets correspond to paths in this graph, which can be interpreted as the problem
of pricing transportation links or network connections. 
(b) a further special case called the {\em highway problem} where the graph is a path,
which can also be motivated from a scheduling perspective (the path corresponds to a
time-horizon). 
The non-SM versions of even such structured problems can be used to capture various 
interesting scenarios. 
%
For instance, in a computer network, users may consider different possibilities
for connecting to the network, and the price they are willing to pay may depend on where
they connect. 
The goal is to determine how to price the bandwidth along the network links so as to
maximize the profit obtained. 
For an application of the non-single minded highway problem, consider 
customers who are interested in executing their jobs on a machine(s) (or using a service,
such as a hotel room). A customer is willing to pay for this service, but the amount
paid depends on when her job is scheduled. 
We want to price the time units and schedule the jobs so as to maximize the
profit. 

\medskip
\noindent {\bf Our results.\ }
We obtain fairly general {\em polytime approximation guarantees}
for profit-maximization problems involving {\em combinatorial auctions} with  
{\em limited supply} and {\em non-single-minded} valuations. 
We obtain results for both (a) certain structured valuation classes,
namely {\em subadditive valuations} (where $v(A)+v(B)\geq v(A\cup B)$ for any two sets $A,B$)  
and hence, {\em submodular} valuations, which have been intensely studied
recently~\cite{DS06,F06,BBM08,V08,CHK09}; and (b) {\em arbitrary valuations}.
Our results relate the approximability of the profit-maximization
problem to that of the corresponding 
{\em social-welfare-maximization} (SWM) problem, which is the problem  of finding an
allocation $(S_1,\ldots,S_n)$ satisfying the capacity constraints that has maximum total
value $\sum_j v_j(S_j)$. 
Our main theorem, stated informally below and proved in Section~\ref{alg}, shows that any
LP-based approximation algorithm that provides an integrality-gap bound for the SWM
problem with a given class of valuations, can be leveraged to obtain a corresponding
approximation guarantee for the profit-maximization problem with that class of valuations.  
Let $c_{\max}\leq n$ denote the maximum item supply, 
and $\OPT_\swm$
denote the optimum value of the SWM problem, which is clearly an upper bound on the
maximum profit achievable.  
\begin{theorem}[{\bf Informal statement}] \label{results}
(i) For the class of subadditive (and hence submodular) valuations, one can
obtain a solution with profit $\OPT_{\swm}/O(\log c_{\max})$. 

\noindent
(ii) Given any class of valuations for which the corresponding SWM problem admits a
packing-type LP relaxation with an integrality gap of $\al$ as  ``verified'' by an
$\al$-approximation algorithm, one can obtain a solution with profit
$\OPT_{\swm}/O(\al\log c_{\max})$. 
\end{theorem}

\noindent
(Part (ii) above does not imply part (i), because for part
(ii) we require an integrality-gap guarantee 
which, roughly speaking, means that we require an algorithm that returns a
``good'' solution for {\em every} profile of $n$ valuations.) 

A key notable aspect of our theorem is its versatility. 
One can simply ``plug in'' various known (or easily derivable) results about the SWM
problem to obtain approximation algorithms for various limited-supply profit-maximization
problems. 
For example, as {\em corollaries} of part (ii) of our theorem, we
obtain an $O(\sqrt m\log c_{\max})$-approximation for profit-maximization for
combinatorial auctions with {\em arbitrary valuations}, and an 
$O(\log c_{\max})$-approximation for the {\em non-single-minded tollbooth problem} on
trees (see Section~\ref{apps}). 
The first result follows from the various known $O(\sqrt m)$-approximation
algorithms for the SWM problem for CAs with arbitrary valuations that
also bound the integrality gap~\cite{P88,KS04}. For the second result, we devise a
suitable $O(1)$-approximation for the SWM problem corresponding to non-single-minded
tollbooth on trees, by adapting the randomized-rounding approach of Chakrabarty et
al.~\cite{CCGK07}. 

Notice that with bundle-pricing, which is often used in the context of mechanism design
for CAs, the profit-maximization problem becomes equivalent to the SWM problem. Thus, 
our results provide worst-case bounds on how item-pricing (which may be viewed as a
fairness constraint on the seller) diminishes the revenue of the seller versus 
bundle-pricing.   
It is also worth remarking that our algorithms for an arbitrary valuation class (i.e.,
part (ii) above) can be modified in a simple way to return prices and an allocation
$(S_1,\ldots,S_n)$ with the following $\e$-``one-sided envy-freeness'' property while
diminishing the profit by a $(1-\e)$-factor: for every non-empty $S_j$, the utility that $j$
obtains from $S_j$ is at least $\e$ times the maximum utility $j$ may obtain from any
set (see Remark~\ref{algprop}). 

The only previous guarantees 
for limited-supply CAs with 
a general valuation-class are those obtained via a reduction in~\cite{AH06},
showing that an $\al$-approximation for the SWM    
problem and an algorithm for the unlimited-supply SM problem that returns profit at least 
$\OPT_\swm/\beta$ yield an $\al\beta$-approximation. 
A simple ``grouping-by-density'' approach gives $\beta=O(\log m+\log n)$; using the
best known bound on $\beta$~\cite{BK06} yields 
an $O\bigl(\al(\log m+\log c_{\max})\bigr)$ guarantee, which is significantly weaker than   
our guarantees. 
(E.g., we obtain an $O(\al)$-approximation for constant $c_{\max}$.) 
The $O(\log c_{\max})$-factor we incur is unavoidable if one compares the
profit against the optimal social welfare: a well-known example with one item,
$n=c_{\max}$ players 
shows a gap of
$H_{c_{\max}}:=1+\frac{1}{2}+\cdots+\frac{1}{c_{\max}}$ between the optima of the 
SWM- and profit-maximization problems. 
Almost all results for profit-maximization for CAs with non-SM
valuations also compare against the optimum social welfare, so they also incur this factor.
Also, it is easy to see that with $c_{\max}=1$, the profit-maximization problem reduces to
the SWM problem, so an inapproximability result for the SWM problem also yields an
inapproximability result for our problem. 
Thus, 
we obtain an $m^{\frac{1}{2}}$-, or $n$-, inapproximability for CAs with even SM valuations
(see, e.g.,~\cite{GLSU06}), and \apx-hardness for CAs with subadditive, submodular 
valuations, and the tollbooth problem on trees.

The proof of Theorem~\ref{results} is based on considering a natural
LP-relaxation \eqref{ca-p} for the SWM problem and its dual. A crucial observation is
that an optimal primal solution combined with the optimal values of the dual variables
corresponding to the primal supply constraints can be seen as furnishing a ``feasible''
solution with a {\em fractional} allocation to (even) the (envy-free) profit-maximization problem.
\cite{CS08} utilized this observation to design an approximation algorithm
for the {\em single-minded envy-free} profit-maximization problem. 
But even with unit capacities and {\em one} non-single-minded customer, 
there is an $\Omega(m)$-factor gap between the optimum (integer or fractional)
social-welfare and the optimum profit achievable by an envy-free pricing  
(see, e.g.,~\cite{BBM08}). Our approach is similar to the one in~\cite{CS08}, but as
suggested by the above fact, we need new ingredients to exploit the greater
flexibility afforded by the profit-maximization problem (vs. the envy-free problem)
and turn the above observation into an approximation algorithm 
{\em even for non-single-minded valuations}. 
As in~\cite{CS08}, we argue that there must be an optimal dual solution with
suitable, possibly lowered, item-capacities yielding profit (with the fractional
allocation) comparable to $\OPT_\swm$. A suitable rounding of the optimal primal
solution with these capacities 
then yields a good allocation, which combined with the prices
obtained yields the desired approximation bounds. 
Here, for part (ii) of Theorem~\ref{results}, we leverage a decomposition technique
of~\cite{CV02}. 
Thus, our work shows that 
(in contrast to the envy-free setting) for profit-maximization problems, one can obtain a
great deal of mileage from the LP-relaxation of the SWM problem 
and exploit LP-based techniques to obtain guarantees 
even for various non-single-minded valuation classes. 



In Section~\ref{s-line}, we consider an alternate approach for the 
non-SM highway problem that (does not use $\OPT_{\swm}$ as an upper bound and)
achieves an (incomparable) $O(\log m)$-approximation factor. 
We decompose the instance via an exponential-size {\em configuration LP}, which is solved 
approximately using the ellipsoid method and rounded via randomized rounding. 
Here, we use LP duality to handle dependencies arising from the non-SM setting. 

\begin{theorem}\label{t-line}
There is an $O(\log m)$-approximation algorithm for the non-single-minded highway problem 
with (i) subadditive valuations with limited supply; and 
(ii) arbitrary valuations with unlimited supply. 
\end{theorem}


It is worth noting that the non-SM highway problem with subadditive valuations 
can be used to capture some \emph{multi-product
  pricing} problems in the so-called {\em Max-Buy} model (a customer buys the most expensive
product she can afford), with or without a \emph{price ladder}, considered by Aggarwal et
al.~\cite{AFMZ04}. (Indeed, the case without a ladder (resp., with a ladder) can be modeled by a set of disjoint intervals (resp., a Laminar set of intervals sharing the right end-point), where customers' valuations are defined on each of these intervals). In fact, our algorithm in Theorem~\ref{t-line} is based on combining ideas from the PTAS for the version with a price ladder in \cite{AFMZ04}, and the $\frac{e}{e-1}$-approximation algorithm for the one without a ladder. 
We observe that our method gives the following result\footnote{This result was recently rediscovered in \cite{FS13}.} for the multi-product pricing problem, which improves on the 2-approximation result in \cite{BK07}.
\begin{corollary}\label{c1}
There is an $\frac{e}{e-1}$-approximation algorithm for the multi-product
  pricing problem in the Max-Buy model, with limited supply. 
\end{corollary}

\medskip
\noindent {\bf Related work.\ }
There has been a great deal of recent work on approximation algorithms for various kinds
of pricing problems; see, e.g.,~\cite{AFMZ04,BBM08,CS08,CHK09,GHKKKM05,BK07, GR11, E12}, 
and the references therein. 
However, to our knowledge, the only approximation results for
profit-maximization for non-single-minded CAs with (general) limited supply (i.e.,
not necessarily unit-  or unlimited- supply) are: (1) those gleaned from the reduction in~\cite{AH06} coupled with the guarantee
in~\cite{BK06}; and 
(2) the 2-approximation algorithm of~\cite{BK07} for {\em unit-demand} valuations (where
each customer wants at most one item). (For this very structured subclass of submodular
valuations, this 2-approximation result is better than the guarantee we obtain using Theorem~\ref{results}; hopwever, our method .) 
 We
briefly survey the work in three special cases that have been studied: unit supply,
unlimited supply, and {single-minded (SM) valuations.} 

As remarked earlier, with unit capacities, the profit-maximization problem reduces to 
the SWM problem, which is a relatively well-studied problem.
The approximation guarantees known for the SWM problem for CAs with unit capacities 
are (i) 2 for subadditive valuations~\cite{F06}; 
(ii) $\frac{e}{e-1}$~\cite{DS06,V08} for submodular
valuations; and (iii) $\Tht(\sqrt m)$ for arbitrary valuations~\cite{P88,KS04}.   
Recently, Balcan et al.~\cite{BBM08} and Chakraborty et al.~\cite{CHK09} considered
the unit-capacity problem with subadditive valuations in the {\em online} setting 
where customers arrive online and select their {\em utility-maximizing}
set from the unallotted items given the current prices. 
The guarantees they obtain in this constrained
setting are naturally worse than the guarantees known in the offline setting. 
The unlimited-supply setting with arbitrary 
valuations has been less studied; \cite{BBM08} 
gave an $O(\log m+\log n)$-approximation algorithm by extending an 
algorithm in~\cite{GHKKKM05} for SM valuations.

The single-minded profit-maximization problem 
has received much attention. The work that is most relevant to ours is Cheung and
Swamy~\cite{CS08}, who obtain an approximation guarantee of the same flavor as in part
(ii) of Theorem~\ref{results}. 
They obtain an envy-free solution of profit $\OPT_\swm/O(\al\log {c_{\max}})$ 
 using an LP-based $\al$-approximation for the SWM problem 
(in the SM setting, this is equivalent to the integrality-gap
requirement we have); 
we make use of portions of their analysis in proving our results.
For the unlimited-supply SM problem,~\cite{GHKKKM05} gave an 
$O(\log m+\log n)$-approximation guarantee, which was improved by Briest and
Krysta~\cite{BK06}. 
A variety of approximation results based on dynamic programming have been
obtained~\cite{GHKKKM05,HK05,BB07,BK06,GLSU06} that yield exact algorithms or
approximation schemes for various restricted instances, or pseudopolynomial or
quasipolynomial time algorithms.   
On the hardness side, a reduction from the set-packing problem shows that achieving an
approximation factor better than $m^{\frac{1}{2}}$, or $n$, is \nphard even when
$c_{\max}=1$, 
even for the (SM) tollbooth problem on grid graphs~\cite{GLSU06}, 
and~\cite{DHFS06,BK06,GHKKKM05,CLN13} prove various hardness results for unlimited-supply
instances. 



Finally, we note that the singe-minded version of the highway problem admits a PTAS in the unlimited supply case \cite{GR11} and a quasi-PTAS for the limited supply case \cite{E12} (with an $\epsilon$-approximate notion of envy-freeness), while the non-single minded
version is APX-hard (since it includes the multi-product pricing problem which was proved
to be APX-hard in the Max-Buying setting in \cite{AFMZ04}).





\section{Problem definition and preliminaries} \label{s-prem}

\paragraph{Profit-maximization problems for combinatorial auctions.}
The general setup of {\em profit-maximization} problems for (multi unit) 
{\em combinatorial auctions} (CAs) is as follows. 
There are $n$ customers and $m$ items.  Let $[n]:=\{1,\ldots,n\}$ and
$[m]:=\{1,\ldots,m\}$. 
Each item $e$ is available in some limited supply or capacity $c_e$. Each customer $j$ has a 
{\em valuation function} $v_j:2^{[m]}\mapsto\R_+$, 
where $v_j(S)$ specifies the maximum amount that customer $j$ is willing to pay for the
set $S$; 
equivalently 
this is $j$'s value for receiving the set $S$ of
items. We assume that $v_j(\es)=0$; we often assume for convenience that 
$v_j(S)\leq v_j(T)$ for $S\sse T$, but this monotonicity requirement is not crucial for
our results.  
The objective is to find non-negative prices $p_e \geq 0$ for the items, and an
allocation $(S_1,\ldots,S_n)$ of items to customers (where $S_j$ could be empty) so as to
maximize the total profit   
$\sum_{j\in [n]} \sum_{e \in S_j} p_e = \sum_{e\in[m]}p_e|\{j: e\in S_j\}|$   
while satisfying the following two constraints.
\begin{list}{$\bullet$}{\usecounter{enumi} \itemsep=-0.3ex \leftmargin=2ex}
\item \textbf{Budget constraints.} Each customer $j$ can afford to buy her assigned set:  
$p\bigl(S_j\bigr):=\sum_{e\in S_j}p_e \leq v_j\bigl(S_j\bigr)$. 

\item \textbf{Capacity constraints.} Each element $e$ is assigned to at most $c_e$
customers:
$|\{j \in [n]: e \in S_j \}| \leq c_e$. 
\end{list}
Since the valuations may be arbitrary set functions, an explicit description of the input
may require exponential (in $m$) space. Hence, we assume that the valuations are
specified via an oracle. As is standard in the literature on combinatorial 
auctions and profit-maximization problems (see, e.g.,~\cite{LS05,F06,BBM08,CHK09}), we assume
that a valuation $v$ is specified by a {\em demand oracle}, 
which means that given item prices $\{p_e\}$, the demand-oracle returns a set $S$ that
maximizes the {\em utility} $v(S)-p(S)$. 
We use $c_{\max}:=\max_e c_e$ to denote the maximum item supply. 

\vspace{-5pt}
\paragraph{An LP relaxation.}
We consider a natural linear programming (LP) relaxation \eqref{ca-p} of the SWM problem
for combinatorial auctions, and its dual \eqref{ca-d}.   
Throughout, we use $j$ to index customers, $e$ to index items, and $S$ to index sets of
items. We use the terms supply and capacity, and customer and player interchangeably. 

\vspace{-10pt}
{\centering \hspace*{-20pt}
\begin{minipage}[t]{.49\textwidth}
\begin{alignat}{3}
\max & \quad & \sum_{j,S} v_j(S) & x_{j,S} \tag{P} \label{ca-p} \\
\text{s.t.} && \sum_{S} x_{j,S} & \leq 1 \quad\ && \forall j \label{asgn} \\
&& \sum_j\sum_{S:e\in S} x_{j,S} & \leq c_e \quad\ && \forall e \label{supply} \\[-4pt]
&& x_{j,S} & \geq 0 && \forall j,S \notag 
\end{alignat}
\end{minipage}
\quad \rule[-24ex]{1pt}{22ex}
\begin{minipage}[t]{0.47\textwidth}
\begin{alignat*}{3}
\min & \quad & \sum_e c_ey_e & +\sum_j z_j \tag{D} \label{ca-d} \\
\text{s.t.} && \sum_{e\in S}y_e+z_j & \geq v_j(S) \quad && \forall j,S \\
&& y_e,z_j & \geq 0 \quad && \forall e,j.
\end{alignat*}
\end{minipage}}

\medskip

\noindent
In the primal LP, we have a variable $x_{j,S}$ for each customer $j$ and set 
$S$ that indicates if $j$ receives set $S$, 
and we relax the integrality constraints on these variables to obtain the LP relaxation.
The dual \eqref{ca-d} has variables $z_j$ and $y_e$ for each customer $j$ and element $e$
respectively, which correspond to the primal constraints \eqref{asgn} and \eqref{supply}
respectively. Although \eqref{ca-d} has an exponential number of constraints, it can be
solved efficiently given demand oracles for the valuations as these oracles yield the
desired separation oracle for \eqref{ca-d}. This in turn implies that \eqref{ca-p} can 
be solved efficiently.
We say that an algorithm $\A$ for the SWM problem is an {\em LP-based $\al$-approximation
algorithm} for a class $\V$ of valuations if for every instance involving valuation
functions $(v_1,\ldots,v_n)$, where each $v_j\in\V$, $\A$ returns an integer solution of
value at least $\text{LP-optimum}/\al$. For example, the algorithm in~\cite{F06} is an LP-based
2-approximation algorithm for the class of subadditive valuations.  


%

\begin{definition} \label{ipverify}
We say that an algorithm $\A$ for the SWM problem {\em ``verifies'' an integrality gap of
(at most) $\al$} for an LP-relaxation of the SWM problem (e.g.,\eqref{ca-p}), if for  
{\em every} profile of (monotonic) valuation functions $(v_1,\ldots,v_n)$, 
$\A$ returns an integer solution of value at least
$(\text{LP-optimum})/\al$.  
\end{definition}

As emphasized above, an integrality-gap-verifying algorithm above must ``work'' for 
every valuation-profile. 
(Note that for the SWM problem, one can always assume that the valuation is
monotonic, since we can always move from a set to its subset (as items may be left
unallotted).) 
In particular, an LP-based $\al$-approximation algorithm for a given 
{\em structured class} of valuations (e.g., submodular or subadditive valuations) 
{\em does not} verify the integrality gap for the LP-relaxation. 
This is the precise reason why our guarantee for subadditive valuations (part (i) of
Theorem~\ref{results}) does not follow from 
part (ii) of Theorem~\ref{results}.   
%
In certain cases however, one may be able to encapsulate the combinatorial structure of
the SWM problem with a structured valuation class by formulating a stronger LP-relaxation
for the SWM problem, and thereby prove that an approximation algorithm for the structured
valuation class is in fact an integrality-gap-verifying approximation algorithm with
respect to this {\em stronger LP-relaxation}. 
For example, in Section~\ref{apps} we consider the
setting where items are edges of a tree and customers desire paths
of the tree. This leads to the structured valuation where, for a set of edges $T$,
$v(T)=\max\{v(P): P\text{ is a path in }T\}$ (with $v(P)\geq 0$ being the value
for path $P$). 
We design an $O(1)$-approximation algorithm for such valuations, and formulate a
stronger LP for the corresponding SWM problem 
for which our algorithm verifies a constant integrality gap.  
(For the SWM problem with subadditive valuations, it is not known how to exploit the
underlying structure 
and formulate an efficiently-solvable LP-relaxation
with $O(1)$ integrality gap.)  


For a given instance $\I=\bigl(m,n,\{v_j\}_{j\in [n]},\{c(e)\}_{e\in [m]}\bigr)$, our
algorithms will consider different capacity vectors $k\leq c$. We use \plp{k} and \dlp{k}
to denote respectively \eqref{ca-p} and \eqref{ca-d} with capacity-vector 
$k=(k_e)$, and $\OPT(k)$ to denote their common optimal value. 
Let $\OPT:=\OPT(c)$ denote the optimum value of \eqref{ca-p} (and \eqref{ca-d})
with the original capacities. 
We will utilize the following facts that follow from complementary slackness, and
a rounding result that follows from the work of Carr and Vempala~\cite{CV02}, and was made
explicit in~\cite{LS05}.  

\begin{myclaim} \label{cscond}
Let $k=(k_e)$ be any capacity-vector, and let $x^*$ and $(y^*,z^*)$ be optimal solutions
to \plp{k} and \dlp{k} respectively. 
\vspace{-1.5ex}
\begin{list}{(\roman{enumi})}{\usecounter{enumi} \itemsep=-0.2ex \leftmargin=2ex}
\item If $x^*_{j,S}>0$, then $\sum_{e\in S}y^*_e\leq v_j(S)$; 
\item If $x^*_{j,S}>0$, and $v_j$ is subadditive, then $\sum_{e\in T}y^*_e\leq v_j(T)$ for
any $T\sse S$;   
\item If $y^*_e>0$, then $\sum_{j,S:e\in S}x^*_{j,S}=k_e$.
\end{list}
\end{myclaim}
\begin{proof}
Parts (i) and (iii) follow directly from the complementary slackness (CS) conditions: part 
(i) follows from the CS condition for $x^*_{j,S}$, since $z^*_i\geq 0$;
part (iii) uses the CS condition for $y^*_e$ and the corresponding primal constraint
\eqref{supply}.
For part (ii), again, by the CS conditions we have $\sum_{e\in S}y^*_e+z^*_j=v_j(S)$. 
Also, dual feasibility implies that $\sum_{e\in S\sm T}y^*_e+z^*_j\geq v_j(S\sm T)$. 
Subtracting this from the first equation and using subadditivity yields 
$\sum_{e\in T}y^*_e\leq v_j(S)-v_j(S\sm T)\leq v_j(T)$.
\end{proof}

\begin{remark} \label{rem-cscond}
As mentioned above, we will sometimes consider a different LP-relaxation 
when considering the SWM problem with a structured class of valuations. 
Roughly speaking, the only properties we require of this LP 
are that it should:
(a) include a constraint similar to \eqref{supply} that encodes the supply constraints; and 
(b) be a {\em packing LP}, i.e., have the form $Ax\leq b,\ x\geq 0$ where $A$ is a
nonnegative matrix. 
Given this, parts (i) and (iii) of Claim~\ref{cscond} continue to hold with $y_e$ denoting
(as before) the dual variable corresponding to the supply constraint for item $e$, since  
the dual is then a covering LP. 
\end{remark}

\begin{lemma}[\cite{CV02,LS05}] \label{decomp}
Given a fractional solution $x$ to the LP-relaxation of an SWM problem that is a
packing LP (e.g., $\plp{k}$), and a polytime integrality-gap-verifying $\al$-approximation 
algorithm $\A$ for this LP, one can express $\frac{x}{\al}$ as a convex combination of
integer solutions to the LP in polytime. In particular, one can round $x$ to a random integer 
solution $\hx$ satisfying the following ``rounding property'':
$\frac{x_{j,S}}{\al}\leq\Pr[\hx_{j,S}=1]\leq x_{j,S}\ \ \forall j,S$. 
\end{lemma}

\section{The main algorithm and its applications} \label{alg}
Claim~\ref{cscond} leads to the simple, but important observation that if $k\leq c$ and
the optimal primal solution $x^*$ is integral, then by using $\{y^*_e\}$ as the prices,
one obtains a feasible solution to the profit-maximization problem with profit 
$\sum_e k_ey^*_e$. 
There are two main obstacles encountered in leveraging this observation and turning it 
into an approximation algorithm. 
First, \plp{k} will of course not in general have an integral optimal solution.
Second, it is not clear what capacity-vector $k\leq c$ to use: for instance, 
$\sum_e c_ey^*_e$ could be much smaller than $\OPT$ (it is easy to construct such
examples), 
and in general, $\sum_e k_ey^*_e$ could be quite small for a given capacity-vector 
$k\leq c$.    
We overcome these difficulties by taking an approach similar to the one 
in~\cite{CS08}. 

We tackle the second difficulty by utilizing a key lemma proved by Cheung and
Swamy~\cite{CS08}, which is stated in a slightly more general form in Lemma~\ref{cheungs}
so that it can be readily applied to various profit-maximization problems. 
This lemma implies that 
one can efficiently compute a capacity-vector $k\leq c$ and an optimal dual
solution $(y^*,z^*)$ to $\dlp{k}$ such that $\sum_e k_ey^*_e$ is 
$\bigl(\OPT-\OPT(\bone)\bigr)/O(\log c_{\max})$, where
$\bone$ denotes the all-one's vector (Corollary~\ref{goodcap}). 
%
To handle the first difficulty, notice that part (i) of Claim~\ref{cscond} implies that
one can still use $\{y^*_e\}$ as the prices, provided we obtain an allocation (i.e.,
integer solution) $\hx$ that only assigns a set $S$ to customer $j$ (i.e., $\hx_{j,S}=1$)
if $x^*_{j,S}>0$. 
(In contrast, in the envy-free setting, 
if we use $\{y^*_e\}$ as the prices then {\em every} customer $j$ with $z^*_j>0$, and hence
$\sum_S x^*_{j,S}=1$, must be assigned a set $S$ with $x^*_{j,S}>0$; this may be
impossible with non-single-minded valuations, whereas this is easy to accomplish with
single-minded valuations (as there is only one set per customer).) 
Furthermore, for subadditive valuations, part (ii) of Claim~\ref{cscond}
shows that it suffices to obtain an allocation where $\hx_{j,T}=1$ implies that there is
some set $S\supseteq T$ with $x^*_{j,S}>0$. This is precisely what our algorithms do.
We show that one can round $x^*$ into an integer solution $\hx$ satisfying the above 
{\em structural} properties, and in addition ensure that the profit obtained, 
$\sum_{j,T}\hx_{j,T}\bigl(\sum_{e\in T}y^*_e\bigr)$, is ``close'' to  
$\sum_e k_ey^*_e$ (Lemma~\ref{round}). 
%
So if $\sum_e k_e y^*_e$ is $\OPT/O(\log c_{\max})$ then applying this rounding procedure 
to the optimal primal solution to $\plp{k}$ yields a ``good'' solution.
On the other hand, Corollary~\ref{goodcap} implies that if this is not the case, then
$\OPT(\bone)$ must be large compared to $\OPT$, and then we observe that an
$\al$-approximation to the SWM problem trivially yields a solution with profit
$\OPT(\bone)/\al$ (Lemma~\ref{unitcap}). 
(As mentioned earlier, in the envy-free setting and unit capacities, there can be an
$\Omega(m)$-gap between the optimum profit 
and the optimum social welfare.) 
Thus, in either case we obtain the desired approximation. 

The algorithm is described precisely in Algorithm~\ref{gen-alg}.
If we use an LP-relaxation different from \eqref{ca-p} for the SWM
problem with a given valuation class that satisfies the properties stated in
Remark~\ref{rem-cscond}, then the only (obvious) change to Algorithm~\ref{gen-alg} is that 
we now use this LP and its dual (with the appropriate capacity-vector) instead of
$\eqref{ca-p}$ and $\eqref{ca-d}$ above. 

\begin{algorithm}[ht!] 
\caption{Non-single-minded profit-maximization} \label{gen-alg}
\algsetup{linenodelimiter=.}
\begin{algorithmic}[1]
\REQUIRE a profit-maximization instance $\I=\bigl(m,n,\{v_j\},\{c_e\}\bigr)$ with a demand
oracle for each valuation $v_j$
\STATE Define $k^1,k^2,\ldots,k^\ell$ as the following capacity-vectors. 
Let $k^1_e=1\ \forall e$. 
For $j>1$, 
let \mbox{$k^j_e=\min\bigl\{\lceil(1+\e)k^{j-1}_e\rceil,c_e\bigr\}$}; let $\ell$ be
the smallest index such that $k^\ell=c$. 
\label{multcap}

\STATE For each vector $k=k^j,\ j=1,\ldots,\ell$, compute an optimal solution 
$(y^\br k,z^\br k)$ to $\dlp{k}$. 
Select $u\in\{k^1,\ldots,k^\ell\}$ that maximizes $\sum_e u_ey^\br u_e$.
\label{bestcap}

\STATE 
Compute an optimal solution $x^\br u$ to $\plp{u}$. 
Use $\round(u,x^\br u)$ to convert $x^\br u$ to a feasible allocation. \label{rounding}  

\STATE Use an LP-based $\al$-approximation algorithm for the SWM problem (with the given
valuation class) to compute an $\al$-approximate solution to the SWM problem with unit 
capacities, and a pricing scheme for this allocation that yields profit equal to the
social-welfare value of the allocation. \label{unitc}

\STATE Return the better of the following two solutions: 
(1) allocation computed in step \ref{rounding} with $\{y^\br u_e\}$ as the prices; 
(2) allocation and pricing scheme computed in step~\ref{unitc}. 
\end{algorithmic}

\renewcommand{\algorithmicloop}{{\boldmath $\round(\mu=(\mu_e),x^*)$}}
\smallskip
\begin{algorithmic}
\LOOP
\STATE {\bf Subadditive valuations:\ }First, independently for each player $j$, assign $j$ 
at most one set $S$ by choosing set $S$ with probability $x^*_{j,S}$. 
If an item $e$ gets allotted to more than $\mu_e$ customers this way, then arbitrarily
select $\mu_e$ customers from among these customers and assign $e$ to these customers. 
Given item prices, this algorithm can be derandomized via the method of conditional
expectations. 

\STATE {\bf General valuation class:\ } Given an integrality-gap-verifying
$\al$-approximation algorithm (for $\plp{\mu}$), use Lemma~\ref{decomp} to decompose
$\frac{x^*}{\al}$ into a convex combination $\sum_{r=1}^\ell \ld_r\hx^r$ of integer
solutions to $\plp{\mu}$. (Here $\sum_r\ld_r=1$ and $\ld_r\geq 0$ for 
each $r$.) Return $\hx^\br r$ with probability $\ld_r$. 
Given item prices, this algorithm can be derandomized by choosing
the \mbox{solution in $\{\hx^\br 1,\ldots,\hx^\br r\}$ achieving maximum profit.}
\ENDLOOP
\end{algorithmic}
\renewcommand{\algorithmicloop}{\textbf{loop}}
\end{algorithm}

\vspace{-5pt}
\paragraph{Analysis.} 
The analysis of Algorithm~\ref{gen-alg} for both subadditive valuations and a general
valuation class proceeds very similarly with the only point of difference being in the 
analysis of the rounding procedure (Lemma~\ref{round}). 
First, observe that if we have an allocation $(S_1,\ldots,S_n)$ that is feasible with unit 
capacities, then since the sets $S_j$ are disjoint we can charge each customer her
valuation for the assigned set by pricing one of her items at this value, and hence,
obtain profit equal to the social-welfare value $\sum_j v_j(S_j)$ of the allocation. 

\begin{lemma} \label{unitcap}
Given an LP-based $\alpha$-approximation algorithm for the SWM problem with a given
valuation class, one can compute a solution that achieves profit at least
$\OPT(\bone)/\alpha$.      
\end{lemma}


\newcommand{\dpr}[1]{\ensuremath{(\mathrm{C}_{#1})}}

\begin{lemma}[\cite{CS08} paraphrased] \label{cheungs}
Let $\dpr{k}$ denote the LP: 
$\min\ k^Ty+b^Tz\ \ \text{s.t.}\ \ (y,z)\in\Pc\sse\R_+^{m+n}$,
where $k,y\in\R_+^m,\ b,z\in\R_+^n,\ \Pc\neq\es$. 
Let $(y^\br k,z^\br k)$ be an optimal solution to $\dpr{k}$ that maximizes $k^Ty$ among 
all optimal solutions, and $\opt(k)$ denote the optimal value.
Let $k^1,\ldots,k^\ell$, and $u$ be as defined in steps~\ref{multcap} and~\ref{bestcap}
respectively of Algorithm~\ref{gen-alg}. Then, 
$\sum_e u_ey^\br u_e\geq\bigl(\opt(c)-\opt(\bone)\bigr)/\bigl(2(1+\e)H_{c_{\max}}\bigr)$.
\end{lemma}
\begin{proof}
We mimic the proof in~\cite{CS08}.
First, note that $\opt(k)$ is well-defined for all $k\geq 0$. 

For $j>1$, define $d^j=k^j-k^{j-1}$. Note that $0\leq d^j_e\leq k^j_e$ for all $e$.
Let $e^*$ be an item with $c_{e^*}=c_{\max}$. 

\begin{myclaim} \label{maxr}
For any $j>1$, we have 
${d^j_{e^*}}/{k^j_{e^*}}=\max_e({d^j_e}/{k^j_e})$ and
${d^j_{e^*}}/{k^{j-1}_{e^*}}=\max_e({d^j_e}/{k^{j-1}_e})$.
\end{myclaim}

\begin{proof} 
It is easy to argue the following by induction on $j$: (i) if $k^{j-1}_e<c_e$ and
$k^{j-1}_{e'}<c_{e'}$, then $k^{j-1}_e=k^{j-1}_{e'}$; and 
(ii) if $c_e\leq c_{e'}$, then $k^j_e\leq k^j_{e'}$ and $d^j_e\leq d^j_{e'}$. 
It is clear that $d^j_e>0$ iff $k^{j-1}_e<c_e$, and that $k^{j-1}_{e*}<c_{e^*}$ for all
$j>1$. Combining these facts, for any $e$ with $d^j_e>0$, we have $k^{j-1}_e<c_e$ and so
$k^{j-1}_e=k^{j-1}_{e^*}$, and since $c_e\leq c_{e^*}$, we have $k^j_e\leq k^j_{e^*}$ and
$d^j_e\leq d^j_{e^*}$.  
Thus, $\frac{d^j_{e^*}}{k^j_{e^*}}\geq\frac{d^j_e}{k^j_e}$ and 
$\frac{d^j_{e^*}}{k^{j-1}_{e^*}}\geq\frac{d^j_e}{k^{j-1}_e}$.
\end{proof}

Let $P=\sum_e u_ey^\br u_e=\max_{k=k^1,\ldots,k^\ell}\sum_e k_ey^\br k_e$.
Then, for every $j>1$, we have
\begin{equation}
P\cdot\frac{d^j_{e^*}}{k^{j-1}_{e^*}}\geq\opt\bigl(k^j\bigr)-\opt(k^{j-1}).
\label{ineq2}
\end{equation}
This follows because, considering $(y,z)=\bigl(y^\br {k^{j-1}},z^\br {k^{j-1}}\bigr)$,
which is a feasible solution for $\dpr{k^j}$ and an optimal solution for
$\dpr{k^{j-1}}$, the RHS is at most 
$\sum_e d^j_ey_e\leq \max_e\frac{d^j_e}{k^{j-1}_e}\cdot\sum_e k^{j-1}_ey_e
\leq \frac{d^j_{e^*}}{k^{j-1}_{e^*}}P$, where the last inequality follows again
from Claim~\ref{maxr} and since $\sum_e k^{j-1}_ey_e\leq P$.

Since $k^j_{e^*}\leq 2(1+\e)k^{j-1}_{e^*}$, we can upper bound the coefficient of $P$ in the above inequality by
$2(1+\eps)\sum_{t=k^{j-1}_{e^*}+1}^{k^j_{e^*}} 1/t$.  
Thus, adding \eqref{ineq2} for all $j>1$ gives 
$P\cdot 2(1+\e)\bigl(H_{c_{\max}}-1\bigr)\geq\opt(c)-\opt(\bone)$. 
\end{proof}

\begin{corollary} \label{goodcap}
The capacity-vector $u$ computed in step~\ref{bestcap} of Algorithm~\ref{gen-alg} satisfies 
the inequality 
$\sum_e u_ey^\br u_e\geq\bigl(\OPT(c)-\OPT(\bone)\bigr)/\bigl(2(1+\e)H_{c_{\max}}\bigr)$.
\end{corollary}

We now analyze the rounding procedure for general and subadditive valuations. 
Together with Lemma~\ref{unitcap} and Corollary~\ref{goodcap}, this yields
Theorem~\ref{mainthm}.  

\begin{lemma} \label{round}
Let $\hx$ be the (random) integer solution returned by procedure \round\ in step 3 of
Algorithm~\ref{gen-alg}. Then $\hx$ combined with the pricing scheme $y^\br u$ is a
feasible solution to the profit-maximization problem with probability 1, which achieves
expected profit at least   
(i) $\bigl(1-\frac{1}{e}\bigr)\sum_e u_ey^\br u_e$ for subadditive valuations; and
(ii) $\sum_e u_ey^\br u_e/\al$ for a general valuation class.
\end{lemma}
\begin{proof}
Feasibility is immediate from Claim~\ref{cscond} since if
a player $j$ is assigned a set $S$ then (i) for a general class of valuations, 
$x^\br u_{j,S}>0$, and (ii) for subadditive valuations, there is some set $T\supseteq S$
such that $x^\br u_{j,T}>0$.
The bound on the profit with a general valuation class follows
from part (iii) of Claim~\ref{cscond} since each item $e$ is assigned to an expected number of
$\bigl(\sum_{j,S:e\in S}x^\br u_{j,S}\bigr)/\al$ players. Note that we only need the  
rounding property in Lemma~\ref{decomp} (and not how it is obtained). 
%
To lower-bound the profit achieved with subadditive valuations, we show that the expected number of players who are allotted an item $f$ is at least $\bigl(1-\frac{1}{e}\bigr)\sum_{j,S:f\in S}x^\br u_{j,S}$ (here, $e$ is the base of the natural logarithm). Notice that
this implies the claim since the expected profit is then at least
$\bigl(1-\frac{1}{e}\bigr)\sum_f y^\br u_f\bigl(\sum_{j,S:f\in S}x^\br u_{j,S}\bigr)=
\bigl(1-\frac{1}{e}\bigr)\sum_f u_fy^\br u_f$ (where the last equality follows from part
(iii) of Claim~\ref{cscond}). 
\medskip

Let $X=(X_{j,S})$ be the random, possibly infeasible
solution computed after the first rounding step.  
Now fix an item $f$. To avoid clutter, we use $x_j$ and $X_j$ below as shorthand for 
$\sum_{S:f\in S}x^\br u_{j,S}$ and $\sum_{S:f\in S}X_{j,S}$ respectively. Also, let
$g_j=\frac{x_j}{u_f}$. 
The expected number of players who are allotted item $f$ after the subsequent ``cleanup'' step is 
\begin{alignat*}{1}
\mathrm{E}[& \min\{u_f, \sum_{j}X_{j}\}]\ 
=\ x_{1}\Bigl(1+\E{\min\{u_f-1,\sum_{j\geq 2}X_{j}\}}\Bigr)+
(1-x_{1})\E{\min\{u_f,\sum_{j\geq 2}X_{j}\}} \\
& \geq\ x_{1}+\E{\min\{u_f,\sum_{j\geq 2}X_{j}\}}
\Bigl(x_{1}\cdot\tfrac{u_f-1}{u_f}+1-x_{1}\Bigr)\ 
=\ x_{1}+\Bigl(1-\tfrac{x_{1}}{u_f}\Bigr)\E{\min\{u_f,\sum_{j\geq 2}X_{j}\}} \\
& \geq\ u_f\Bigl(g_{1}+(1-g_{1})g_{2}+\cdots+(1-g_{1})\ldots(1-g_{n-1})g_{n}\Bigr)\ 
=\ u_f\Bigl(1-\prod_{j=1}^n(1-g_j)\Bigr)
\end{alignat*}
Thus,
{$\E{\min\{u_f,\sum_{j}X_{j}\}} \geq 
u_f\bigl(1-\bigl(1-\tfrac{\sum_j g_j}{n}\bigr)^n\bigr)
\geq u_f\bigl(1-(1-\tfrac{1}{n})^n\bigr)\sum_j g_j 
\geq \bigl(1-\tfrac{1}{e}\bigr)\sum_j x_j$.}
Here the penultimate inequality follows from the fact that $1-\bigl(1-\frac{a}{n}\bigr)^n$
is a concave function of $a$, and hence is at least
$\bigl(1-(1-\frac{1}{n})^n\bigr)a$ when $a\in[0,1]$ (note that $\sum_j g_j\leq 1$). 
\end{proof}

\vspace{-5pt}
\begin{theorem} \label{mainthm}
Algorithm~\ref{gen-alg} runs in time $\poly\bigl(\text{input size},\frac{1}{\e}\bigr)$ and
achieves an
\vspace{-1.5ex}
\begin{list}{(\roman{enumi})}{\usecounter{enumi} \itemsep=-0.5ex 
\settowidth{\labelwidth}{\em (ii)} \leftmargin=\labelwidth 
\addtolength{\leftmargin}{\labelsep} 
} 
\item $O(\log c_{\max})$-approximation for subadditive valuations, using the 
2-approximation algorithm for the SWM problem with subadditive valuations 
in~\cite{F06}; 

\item $O(\al\log c_{\max})$-approximation for a general valuation class given an
integrality-gap-verifying $\al$-approximation algorithm for the SWM problem.
\end{list}  
\end{theorem}

\begin{proof}
By Lemmas~\ref{unitcap}, \ref{round}, and Corollary~\ref{goodcap}, for subadditive
valuations, the profit obtained is at least  
$
\max\bigl\{\tfrac{\OPT(\bone)}{2},(1-\tfrac{1}{e})(\OPT(c)-\OPT(\bone))/(2(1+\e)H_{c_{\max}})\bigr\}
\geq\OPT(c)/\bigl(4(1+\e)H_{c_{\max}}\bigr)$. 
Similarly for a general valuation class, we obtain profit at least
$
\frac{1}{\al}\cdot\max\bigl\{\OPT(\bone),(\OPT(c)-\OPT(\bone))/(2(1+\e)H_{c_{\max}})\bigr\}
\geq\OPT(c)/\bigl(4\al(1+\e)H_{c_{\max}}\bigr)$. 
\end{proof}

\begin{remark} \label{algprop}
Note that if the allocation $(S_1,\ldots,S_n)$ returned by Algorithm~\ref{gen-alg} is
obtained via $\round$, then $S_j$ is always a subset of a {\em utility-maximizing set} of
$j$, and with a general valuation class, if $S_j\neq\es$, it is a utility-maximizing set
(under the computed prices). 
(For submodular valuations, this implies that 
$v_j(S_j)-v_j(S_j\sm\{e\})\geq\text{(price of $e$)}$ for all  
$e\in S_j$.) 
If $(S_1,\ldots,S_n)$ is obtained in step~\ref{unitc}, then we may assume that 
$v_j(S_j)=\max_{T\sse S_j}v_j(T)$ (since we have a demand oracle for $v_j$); 
with a general valuation class, this solution 
can be modified to yield an approximate ``one-sided envy-freeness''
property. 
We compute $(S_1,\ldots,S_n)$ by rounding $x^\br 1$ as described in
Lemma~\ref{decomp}. 
Now choose prices $\{p'_e\}$ (arbitrarily) such that $p'\geq y^\br 1$ and
$p'(S_j)=\max\{y^\br 1(S_j),(1-\e)v_j(S_j)\}$ for every $j$.
Since any non-empty $S_j$ is a
utility-maximizing set under $y^\br 1$, it follows that (a) $p'$ is a valid item-pricing
yielding profit at least $(1-\e)\sum_j v_j(S_j)$; (b) if $S_j\neq\es$, then the utility $j$
derives from $S_j$ under $p'$ is at least $\e(\text{max utility of $j$ under $p'$})$. 

These properties prevent a kind of ``cheating'' that may occur in 
profit-maximization problems. 
To elaborate, although monotonicity of the valuation is an innocuous assumption for the SWM
problem, with profit-maximization this can lead to the following artifact: a
customer $j$ desires a set $A$ but is allotted $B\supseteq A$ (with $v_j(B)=v_j(A)$) 
and items in $A$ have 0 price and items in $B\sm A$ have positive prices, so that $j$
ends up paying for items she never wanted! 
The above properties ensure that (we may assume that) 
{\em the solution computed by our algorithm does not have this artifact}. In fact, if
$j$ desires one of $k$ sets $A_1,\ldots,A_k$, then our algorithm will assign $j$ a set
$S_j\in\{\es,A_1,\ldots,A_k\}$. We could also prevent this artifact by dropping
monotonicity of the valuations. 
\end{remark}

\vspace{-20pt}
\subsection{Applications} \label{apps}

\vspace{-5pt}
\paragraph{Arbitrary valuation functions.}
The integrality gap of \eqref{ca-p} is known to be $\Tht(\sqrt m)$, and there are
efficient (deterministic) algorithms that verify this integrality
gap~\cite{P88,KS04}. So Theorem~\ref{mainthm} immediately yields an
$O(\sqrt m\log c_{\max})$-approximation algorithm for the profit-maximization problem for 
combinatorial auctions with arbitrary valuations.

\vspace{-5pt}
\paragraph{Non-single-minded tollbooth problem on trees.} 
In this profit-maximization problem, 
items are \emph{edges} of a
tree and customers desire paths of the tree. More precisely, let $\Pc$ denote the set
of all paths in the tree (including $\es$). Each customer $j$ has a value $v_j(S)\geq 0$
for path $S\in\cP$, and may be assigned any (one) path of the tree. Notice that this leads
to the  {\em structured} valuation function $v_j:2^{[m]}\mapsto\R_+$ where 
$v_j(T)=\max\{v_j(S): S\text{ is a path in }T\}$. Note that $v_j$ need {\em not} be
subadditive. 
We use Algorithm~\ref{gen-alg} to obtain an $O(\log c_{\max})$-approximation 
guarantee by formulating an LP-relaxation of the SWM problem that is tailored to this
setting and designing an $O(1)$-integrality-gap-verifying algorithm for this LP. 

The ``new'' LP is almost identical to \eqref{ca-p}, except that we now 
{\em only have variables $x_{j,S}$ for $S\in\Pc$}.
Correspondingly, in the dual \eqref{ca-d}, we only have a constraint for $(j,S)$ when
$S\in\Pc$. Clearly, this new LP satisfies the properties stated in
Remark~\ref{rem-cscond}, so parts (i) and (iii) of Claim~\ref{cscond} hold for this new
LP, and so does Lemma~\ref{decomp}.
Thus, we only need to design an $O(1)$-integrality-gap-verifying algorithm for this new
LP to apply Theorem~\ref{mainthm}.
Let $\{v_j:\Pc\mapsto\R_+\}_{j\in[n]}$ be any instance and $x^*$ be an optimal
solution to this new LP for this instance. We design a randomized algorithm that returns a 
(random) integer solution $\hx$ of expected objective value
$\Omega(\sum_{j,S\in\Pc}v_j(S)x^*_{j,S})$. 
This algorithm can be derandomized using the work of~\cite{Sivakumar02}; 
this yields an $O(1)$-integrality-gap-verifying algorithm for the new LP.
(We have not attempted to optimize the approximation factor.)
Our algorithm is a generalization of the
one proposed by~\cite{CCGK07} for unsplittable flow on a line.  
Root the tree at an arbitrary node. Define the \textit{depth}
of an edge $(a,b)$ to be the minimum of the distances of $a$ and $b$ to the root. 
Define the depth of an edge-set $T$ to be the minimum depth of any edge in $T$. 
Let $\alpha=0.01$. 
\begin{list}{\arabic{enumi}.}{\usecounter{enumi} \itemsep=0ex \leftmargin=3ex}
\item 
Independently, for every customer $j$, choose at most one set (i.e., path) $S$, by picking
$S$ with probability $\alpha x^*_{j,S}$. Let $S_j$ be the set assigned to $j$. 
(If $j$ is unassigned, then $S_j=\es$.) 
 
\item 
Let $W=\emptyset$. Consider the sets $\{S_j\}$ in non-decreasing order of their
depth (breaking ties arbitrarily). For each set $T=S_j$, if $T$ can be added to 
$\{S_i:~i\in W\}$ without violating any capacities, add $j$ to $W$; otherwise
discard $T$.    
\end{list}
Let $\hx$ be the (random) integer solution computed.
Using a similar argument as in \cite{CCGK07}, we prove in Appendix~\ref{append-apps} that
if we select $\alpha=0.01$, then $\Pr[\hx_{j,S}=1]\geq 0.00425 x^*_{j,S}$, so 
$\E{\sum_{j,S\in\Pc}v_j(S)\hx_{j,S}}\geq 0.00425\cdot\sum_{j,S\in\Pc}v_j(S)x^*_{j,S}$. 
We thus obtain the following theorem as a corollary of Theorem~\ref{mainthm}. 

\begin{theorem}
\label{thm:envyfree_pricing_on_trees}
There is an $O(1)$-integrality-gap-verifying algorithm (for the new LP mentioned above).
This yields an $O(\log c_{\max})$-approximation algorithm for the non-single-minded
tollbooth problem on trees.
\end{theorem}

We remark that since the above algorithm satisfies the rounding property in
Lemma~\ref{decomp}, we can directly use it to round $x^\br u$ (more efficiently)
to a feasible allocation in step~\ref{rounding} of Algorithm~\ref{gen-alg}, instead of 
using the Carr-Vempala decomposition procedure (which relies on the ellipsoid method). 

\section{Refinement for the non-single-minded highway problem} \label{s-line} 

In this section, we describe a different approach 
that does not use $\OPT_{\swm}$ 
as an upper bound on the optimum profit. Instead our approach is based on 
using an exponential-size {\em configuration LP} 
to decompose the original instance into various
smaller (and easier) instances. We use this 
to obtain an 
$O(\log m)$-approximation for the non-single-minded (non-SM) highway problem (recall that
this is the tollbooth problem on a path, so customers desire intervals) with subadditive 
valuations, and arbitrary valuations but unlimited supply (Theorem~\ref{t-line}). Note
that this is {\em incomparable} to the $O(\log n)$-approximation obtained earlier for the
tollbooth problem on trees (as $c_{\max}\leq n$); the number of distinct sets is $O(m^2)$
but the number of customers can be much larger (or smaller). Also, 
an $O(\log m)$-approximation is impossible to obtain using the approach in
Section~\ref{alg}, and in general any approach that uses the optimum of the (integer or
fractional) SWM problem as an upper bound, because, as mentioned earlier, there is a
simple example with just one item and $c_{\max}=n$, 
where the SWM-optimum is an $H_{c_{\max}}$-factor away from the optimum profit.    
Let $\Pc$ be the set of all intervals on the line (with $m$ edges).
As with the non-SM tollbooth problem on trees (in Section~\ref{apps}), each
customer $j$ has a value for each subpath (which is now an interval). So we view $v_j$
as a function $v_j:\Pc\mapsto\R_+$, and {\em subadditivity} means that for any two intervals
$A$, $B$, where $A\cup B$ is also an interval, we have $v_j(A\cup B)\leq v_j(A)+v_j(B)$. 


We outline the proof of Theorem~\ref{t-line}.
First, we use a simple procedure (Proposition~\ref{pro:decomposition}) to partition the
intervals into $O(\log m)$ disjoint sets, where each set is a union of item-disjoint 
``cliques''. Here, a clique is a set of paths that share a common edge; two cliques
$\Pc_1$ and $\Pc_2$ are item-disjoint, if $A\cap B=\es$ for all $A\in\Pc_1,\ B\in\Pc_2$.  

\begin{proposition}[see~\cite{BFNW02}]
\label{pro:decomposition}
A set of $k$ intervals on the line can be partitioned into at most $\lfloor\log
(k+1)\rfloor$ sets, each of which is a union of item-disjoint cliques.
\end{proposition}

\noindent
Thus, we can decompose $\Pc$ into $O(\log m)$ sets; to get an 
$O(\log m)$-approximation algorithm, it suffices to give an $O(1)$-approximation 
algorithm when the intervals form a union of item-disjoint cliques. 
It is unclear how to achieve a near-optimal solution even in this structured setting,
as there are various {\em dependencies between the cliques} in a set: 
a customer can only be assigned an interval in {\em one} of the cliques. 
We solve this ``union-of-cliques'' pricing problem as follows. 
We first trim each clique $\Pc_i$ in our set randomly to a one-sided half-clique by
(essentially) ignoring the items to the left or right of the common edge of
$\Pc_i$. The details of this truncation are slightly different depending on whether we
have subadditive or arbitrary valuations (see the proof of
Lemma~\ref{lem:clique-pricing}), but a key 
observation is that, in expectation, we only lose a factor of 2 by this truncation. 
We formulate an LP-relaxation for the pricing problem involving these half-cliques.  
Solving this LP requires the ellipsoid  method, where the separation oracle is provided by
the solution to another (easier) pricing problem, where the {\em (half) cliques are now
decoupled}. We devise an algorithm based on dynamic programming (DP) to compute a
near-optimal solution to this pricing problem, which then yields a near-optimal solution
to the LP (Lemma~\ref{lem:approx_solution}). Finally, we argue that this near-optimal fractional
solution can be rounded to an integer solution losing only an $O(1)$-factor
(Lemma~\ref{lem:random_rounding}). Combining the various ingredients, we obtain the desired
$O(1)$-approximation for the ``union-of-cliques'' pricing problem, which in turn yields
an $O(\log m)$-approximation for our \nolinebreak \mbox{original non-single-minded highway problem.}


We assume in the following that the edges of the line are numbered $1,2,\ldots,m$, from left to right.
\begin{lemma}
\label{lem:clique-pricing}
There is a $16(1+\frac{1}{m})$-approx. algorithm for the non-SM highway
problem 
when intervals form a union of item-disjoint cliques for (i) subadditive
valuations with limited supply; (ii) arbitrary valuations with unlimited supply. 
\end{lemma}

%
%
\begin{proof}
Let $\mathcal{A} = \bigcup_{i} \cP_i$ be a set of intervals where the $\cP_i$s are
item-disjoint cliques. 
Let $e_i$ denote the common edge of $\cP_i$, and $\ell_i$ and $r_i$ be the leftmost and
rightmost edge used by some interval of $\Pc_i$. 
We first trim the cliques to one-sided half-cliques. 
For every clique $\Pc_i$ independently, we discard one of the ``halves'' of $\Pc_i$
with probability $1/2$. More precisely, for subadditive valuations, discarding the right
half 
means that we truncate each interval
$S\in\Pc_i$ to $S\cap[\ell_i,e_i]$ to obtain the half-clique $\Hc_i$ of truncated
intervals; when discarding the left half we set $\Hc_i=\{S\cap[e_i+1,r_i]: S\in\Pc_i\}$.
For arbitrary valuations with unlimited supply, discarding the right half is defined to
simulate the effect of pricing all edges in $[e_i+1,r_i]$ at 0 (discarding
the left half is symmetric). So in this case, we define the half-clique $\Hc_i$ to be 
$\{S\cup [e_i+1,r_i]: S\in\Pc_i\}$ (note that there are no capacity constraints). 

A key observation is that for both subadditive and arbitrary valuations, 
$\E{\opt(\Hc_i)}\geq\opt(\Pc_i)/2$ for every $i$, where $\opt(\Sc)$ denotes the optimum
profit when players may only be assigned intervals from $\Sc$. 


We now consider the problem of setting {\em interval prices} for the intervals in
$\bigcup_i\Hc_i$ that of course obey the constraint that $p(S)\leq p(T)$ if $S\sse T$. 
First, we discretize the space of interval prices. 
Let $\mathcal{B}$ be the maximum price any player may pay in a feasible solution. 
We consider only positive prices of the form 
$d_q=\mathcal{B}/2^{q},\ q\in\Z_{\geq 0}$ for $d_q \geq \mathcal{B}/mn$. 
We lose at most a factor of $2(1+\frac{1}{m})$ this way (since we have item-disjoint
half-cliques). Now we have $O(\log n+\log m)$ different prices.  
Let $\mathcal{R}_i$ denote the set of all possible solutions for $\Hc_i$, where a solution
specifies a pricing of the intervals in $\Hc_i$ (choosing non-zero prices from $\{d_q\}$ or $0$)
and an allocation of intervals to customers satisfying the budget and capacity constraints.
We introduce a variable $y_{jp}\geq 0$ for each customer $j$ and price $p$ denoting if 
customer $j$ buys a path at price $p$, and a variable $x_{i,R}$ for each 
$R \in \mathcal{R}_i$ 
denoting whether solution $R$ has been chosen for $\Hc_i$. 
Let $p_j(R)$ be the price that $j$ pays under the solution $R$, and
$\mathcal{R}_{i,j,p}=\bigl\{R\in\Hc_i: p_j(R)=p\bigr\}$
be the set of solutions for $\Hc_i$ where $j$ pays price $p$. 
We consider the following LP. Here $p$ indexes all the possible interval-prices.
\begin{alignat}{3}
\max & \quad & \sum_{j,p} p\cdot y_{jp} & \tag{P2} \label{primal_unlimited} \\
\text{s.t.} && \sum_{R \in \cR_i} x_{i,R} & = 1 \qquad && \frall i \notag \\
&& \sum_{p} y_{jp} & \leq 1 \qquad && \frall j \label{single_product} \\
&& \quad y_{jp} & \leq \sum_{i,R: R \in \mathcal{R}_{i,j,p}} x_{i,R} 
\qquad && \frall j, p \label{primal_constraints} \\
&& x_{i,R}, y_{jp} & \geq 0 \qquad && \frall i, R, j, p. \notag
\end{alignat}
Constraint \eqref{single_product} ensures that a customer only buys at at most one price, and 
constraint \eqref{primal_constraints} ensures that $j$ can only buy at price $p$ if a
solution $R \in\bigcup_i\cR_{i,j,p}$ has been selected.   
The arguments above establish that $\OPT_{\text{\eqref{primal_unlimited}}}$ is at
least $1/4\bigl(1+\frac{1}{m}\bigr)$-fraction of the optimum for the instance $\A$ 
(for both subadditive and arbitrary valuations).
We show that one can obtain an integer solution to \eqref{primal_unlimited} of objective
value at least $\OPT_{\text{\eqref{primal_unlimited}}}/4$; this will complete the proof.

\eqref{primal_unlimited} has an exponential number of variables, so to solve it we
consider the dual problem. The separation oracle for the dual amounts to solving a related
pricing problem where the half-cliques are now decoupled. We give a $2$-approximation
algorithm for this problem, which then yields a $2$-approximate dual solution, and hence,
a $2$-approximate solution to \eqref{primal_unlimited} (Lemma~\ref{lem:approx_solution}). 
Lemma~\ref{lem:random_rounding} states that this fractional solution can then be
rounded to an integer solution losing at most another factor of 2. 
This completes the proof. 
\end{proof}

\begin{lemma}
 \label{lem:approx_solution}
One can compute a $2$-approximate solution to \eqref{primal_unlimited} in
polynomial time. 
\end{lemma}

\begin{proof} 
We first show how to get a $2$-approximate solution for the dual
problem. This will also yield a method to get a $2$-approximate primal solution. Consider 
the dual program: 
\begin{alignat}{3}
\min & \quad & \sum_{i} \alpha_i & + \sum_{j} \beta_j \label{dual_unlimited} \tag{D2}\\
\text{s.t.} && \beta_j + \gamma_{jp} & \geq p  \qquad && \text{for all $j$, price $p$} 
\label{dual_ineq1} \\
&& \sum_{j,p: R \in \cR_{i,j,p}} \gamma_{jp} & \leq \alpha_i \qquad && 
\text{for all $i, R \in \cR_i$} \label{clique_constraints} \\
&& \beta_j, \gamma_{jp} & \geq 0 \qquad && \text{for all $j$, price $p$}.
\label{dual_nonnegative} 
\end{alignat}
Note that \eqref{dual_unlimited} has an exponential number of constraints. In order to
%
%
solve \eqref{dual_unlimited} efficiently, we use the ellipsoid method, which reduces the
problem of solving the LP to finding a separation oracle that, given a candidate solution
vector $v = (\alpha,\beta,\gamma)$, either produces a feasible solution (usually the input
$v$) or returns a constraint violated by $v$.  
Constraints \eqref{dual_ineq1} and \eqref{dual_nonnegative} can easily be checked
in polynomial time. In the following, we will therefore assume that $v$
satisfies these constraints. On the other hand, there are an exponential number of
constraints \eqref{clique_constraints}. It turns out that checking whether one of these
constraints is violated amounts to solving a generalized non-single-minded pricing
instance: we seek to find a solution $R \in \cR_i$, such that 
$\sum_{j,p: R \in \cR_{i,j,p}} \gamma_{jp}$ is maximized, i.e., a customer  $j$ who is
allotted an interval priced at $p$ contributes $\gamma_{jp}$ to the objective value.  
If the maximum achievable profit on some half-clique $\Hc_i$ is larger than $\alpha_i$, then
the corresponding solution yields a violated constraint. Otherwise all constraints are
satisfied. Although it seems that we have not gained much by reducing the original
non-single-minded pricing problem to another even more generalized pricing problem, the
crucial point here is that we have \emph{removed the dependencies} between different
(half) cliques. 
Instead of solving this problem with the given $\gm$, we move to a slightly more
structured $\tilde{v} = (\alpha,\beta,\tilde{\gamma})$ to be specified shortly, such that if
$\tilde{v}$ violates a constraint, then the same constraint is also violated by $v$. Note
also that $\tilde{v}$ and $v$ have the same objective function value since they share the
same $\alpha$ and $\beta$ values. 
We define $\tilde{\gamma}_{jp} := \max\{0, p-\beta_j\}$. Note that $\gamma \geq
\tilde{\gamma}$ since we have assumed that $v$ satisfies the constraints
\eqref{dual_ineq1} and \eqref{dual_nonnegative}. Hence, if $\tilde{v}$ violates one of the
constraints \eqref{clique_constraints}, $v$ also violates it. 

We give a $2$-approximation algorithm for this new pricing problem on a half-clique,
which we call the {\em voucher-pricing problem}. (The rationale is that $\beta_j$ can be
viewed as a voucher that customer $j$ can redeem and thereby decrease her price).
We describe the algorithm shortly, but first we show that this implies the lemma. 
Using this approximation algorithm for the separation oracle yields a dual solution 
$(\alpha, \beta, \gamma)$ that is potentially not feasible. In particular, it could violate the
constraints \eqref{clique_constraints}, however only by a factor of at most $2$. By
scaling $\alpha$ accordingly, we get a feasible solution $(2\alpha, \beta, \gamma)$ whose
objective function value is at most $2\cdot\OPT_{\text{\eqref{dual_unlimited}}}$.  

Now applying an argument similar to the one used by Jain et al.~\cite{JainMS03} shows that  
one can also compute a $2$-approximate primal solution.
\end{proof}  

\vspace{-1ex}
\paragraph{A \boldmath $2$-approximation algorithm for the voucher-pricing problem on a
half-clique $\Hc_i$.} 
The algorithm follows the dynamic-programming approach by Aggarwal et
al.~\cite{AFMZ04}. The main observation is that if we relax the constraint that a customer
buys (i.e., is assigned) at most one interval and only prevent her from buying two
intervals at the same  price then we lose at most a factor of $2$. To see
this, note that since we have discretized our search space of prices, if $p$ is the
maximum price paid by a customer $j$ in a solution to the relaxed problem (where she can
buy multiple intervals), then $j$'s contribution to the profit is at most
\begin{align*}
\sum_{q \geq 0} \max\{p\cdot 2^{-q}-\beta_j,0\} & 
\leq \sum_{q \geq 0} 2^{-q} \cdot\max\{p-\beta_j,0\} && \text{since $\beta_j \geq 0$} \\ 
& \leq 2\cdot\max\{p-\beta_j,0\}
\end{align*}
and if we assign $j$ only the single interval at price $p$ (note that this still satisfies
the capacity constraints), we get profit $\max\{p-\beta_j,0\}$.

We solve this relaxed problem using dynamic programming. 
To keep notation simple, let $T_1\supseteq T_2\supseteq\ldots\supseteq T_\ell$ denote the
intervals in $\Hc_i$. 
So we require that $p(T_1)\geq\ldots\geq p(T_\ell)$. 
Let $d_Q$ be the lowest non-zero price in our discrete price-space, and let $d_{Q+1}:=0$. 

Let $F(q,i,U)$ denote the value of an optimal solution when customers are only assigned
intervals from from $\{T_i,\ldots,T_\ell\}$, the prices of these intervals lie in
$\{d_q,\ldots,d_Q,d_{Q+1}\}$, and 
$U$ customers have been assigned intervals from $\{T_1,\ldots,T_{i-1}\}$ (recall that a
customer may be assigned multiple intervals if they are priced differently).
Clearly $F(0,1,0)$ is the optimal value we are looking for.
The base cases are easy: we set $F(Q+1,i,U)=0$ for all $i,\ U$.
For $k\geq i$, let $C(i,k,q)$ denote the set of customers who can afford to buy an
interval in $\{T_i,\ldots,T_k\}$ priced at $d_q$. 
Set $C(i,k,q) = \emptyset$ for $k<i$.   

Suppose we decide to set the price of intervals $T_i,\ldots,T_k$ to $p=d_q$ and assign $t$
customers to these intervals. Then, the best value that one can earn from intervals
$\{T_{k+1},\ldots,T_\ell\}$ is $F(q+1,k+1,U+t)$. Notice that this {\em does not} depend on
which $t$ customers are assigned intervals in $T_i,\ldots,T_k$ or how these customers are
allotted these intervals. Thus, we can compute the optimum assignment of intervals in
$T_i,\ldots,T_k$ to $t$ customers separately. This is an interval packing problem
which one can solve efficiently. 
For $j\in C(i,k,q)$, let $i_j$ be the largest index $i'\in\{i,\ldots,k\}$ such that $j$
can afford to buy $T_{i'}$ at price $p=d_q$. Notice that we may assume that in 
an optimal solution to this interval packing solution, if $j$ is assigned an interval,
it is assigned $T_{i_j}$.   
Now we can formulate the following integer program for solving this interval packing
problem. For each $j\in C(i,k,q)$, let $Z_j$ be an indicator variable that denotes if
$j$ is assigned interval $T_{i_j}$. 
Then, we want to solve the following integer program.
$$ \max \ \sum_{j\in C(i,k,q)}\!\!\!\!Z_j\max\{d_q-\beta_j,0\} \quad
\text{s.t.}\quad \sum_j Z_j=t, \quad \sum_{j:e\in T_{i_j}}Z_j\leq c_e-U \ \forall e, 
\quad Z_j\in\{0,1\}\ \forall j. $$
It is well known that an interval packing problem can be solved efficiently (e.g., by
finding an optimal solution to its LP-relaxation). Let $P(i,k,q,t)$ be the optimal value
of the above program, and $-\infty$ if the program is infeasible.
Then we have the following recurrence.
$$
F(q,i,U) = \max \Bigl\{F(q+1,k+1,U+t)+P(i,k,q,t): \quad
i-1\leq k\leq\ell,\ 0\leq t\leq n\Bigr\}.
$$
We need to compute $O\bigl((\log m+\log n)\ell n\bigr)$ table entries for
$F(\cdot,\cdot,\cdot)$ to get the optimal value, and so this DP can be implemented
in polynomial time.  
Note that we can easily record the corresponding solution along with the computation of 
each $F(q,i,U)$.   

\begin{lemma}
\label{lem:random_rounding}
One can round any solution $(x,y)$ to \eqref{primal_unlimited} to an integer solution of
objective value at least $(1-e^{-1})\sum_j py_{jp}$.
\end{lemma}

\begin{proof} 
Let $(x,y)$ denote a feasible solution to \eqref{primal_unlimited}. 
Recall that $p_{j}(R)$ is the price $j$ pays in the solution $R$. 
For each $j$ and each solution $R\in\bigcup_i\Rc_i$ we can define values $y_{ij,R}$ such
that $y_{ij,R}\leq x_{i,R}$ for all $i,j,R\in\Rc_i$, and 
$y_{jp}=\sum_{i,R\in\Rc_{i,j,p}}y_{ij,R}$ for all $j,p$. By making ``clones'' of a solution
$R\in\Rc_i$ if necessary, we can ensure that $y_{ij,R}$ is either 0 or $x_{i,R}$ for every
$R\in\Rc_i$. 
The rounding is simple: independently, for each half-clique $\Hc_i$, we choose solution 
$R\in\mathcal{R}_i$ with probability $x_{i,R}$. 
Let $Q_i$ denote the solution selected for $\Hc_i$. 
Now we assign each customer $j$ to the $\Hc_i$ with maximum $p_{j}(Q_i)$. Notice
that this yields a feasible solution for the instance composed of the union of the (half)
cliques. The analysis is quite similar to the analysis in~\cite{DS06} and~\cite{AFMZ04}; we
reproduce it here for completeness.

Fix a customer $j$.
Let $\tht_i=\Pr[y_{ij,Q_i}>0]=\sum_{R\in\Rc_i}y_{ij,R}$.
Let $Z_{ij}=\bigl(\sum_{R\in\Rc_i}y_{ij,R}p_{j}(R)\bigr)/\bigl(\sum_{R\in\Rc_i}y_{ij,R}\bigr)$. 
Note that $\sum_{i}Z_{ij}\tht_i=\sum_{i}\sum_{R\in\Rc_i}y_{ij,R}p_{j}(R)=\sum_p py_{jp}$. 
Consider the sub-optimal way of assigning $j$ to a (random) $R$, where we assign $j$ to
the $\Hc_i$ with maximum $Z_{ij}$ for which $y_{ij,Q_i}>0$; if there is no such $i$ then
$j$ is unassigned. Let $k$ be the number of half-cliques, and let these be ordered so that 
$Z_{1j}\geq Z_{2j}\geq\ldots Z_{kj}$. 
The expected price that $j$ pays under this suboptimal assignment is
$$
\tht_1 Z_{1j}+(1-\tht_1)\tht_2 Z_{2j}+\ldots+(1-\tht_1)\cdots(1-\tht_{k-1})\tht_k Z_{kj}
$$
which following the analysis in~\cite{DS06} (for example) is at least
$\bigl(1-\frac{1}{e}\bigr)\sum_i Z_{ij}\tht_i=\bigl(1-\frac{1}{e}\bigr)\sum_{j,p} py_{jp}$.
Thus, the expected profit obtained is at least $(1-e^{-1})\sum_p py_{jp}$.
The algorithm can be derandomized using a simple pipage rounding argument.
\end{proof}

\begin{proofof}{Corollary~\ref{c1}}
As we mentioned in the introduction, the Max-Buy multi-product pricing problem can be viewed as a non-SM highway problem, where there are $m$ disjoint edges and each bidder has values only on theses edges. As we did in the proof of Lemma~\ref{lem:clique-pricing}, we write an LP of the form~\eqref{primal_unlimited}, except that we do not have to round the prices (note that the number of relevant positive prices on each edge is at most the number of customer valuations on that edge).  Then following the approach we used in the proof of Lemma~\ref{lem:random_rounding}, we can solve the dual problem~\eqref{dual_unlimited} in polynomial time since a separation oracle reduces to the voucher pricing problem on a {\it single} item which can be trivially solved in polynomial time. 
  
The claim then follows immediately by combining this with the rounding procedure of Lemma~\ref{lem:random_rounding}. 
\end{proofof}

\noindent
{\large\bf Acknowledgment\ }We thank Markus Bl\"{a}ser for his valuable comments.




\appendix




\section{Proof of Theorem \ref{thm:envyfree_pricing_on_trees}} \label{append-apps}
We prove that the rounding algorithm described in Section~\ref{apps}, 
``Non-single minded tollbooth problem on trees'' is an $O(1)$-integrality-gap verifying
algorithm for the new LP relaxation. 
Recall that this new LP-relaxation for the SWM problem is derived from \eqref{ca-p} by
retaining only variables $x_{j,S}$ where $S$ is a path of the tree. 
More specifically, we prove that with $\al=0.01$ (as defined in the algorithm), the
algorithm returns a random integer solution $\hx$ such that 
$\Pr[\hx_{j,S}=1]\geq 0.00425 x^*_{j,S}$. This implies that
$\E{\sum_{j,S\in\Pc}v_j(S)\hx_{j,S}}\geq 0.00425\cdot\sum_{j,S\in\Pc}v_j(S)x^*_{j,S}$. (Recall
that $\Pc$ is the collection of all paths of the tree.)  
Denote by "$\leq$" the ordering defined in step 2 of the rounding procedure. This ordering
has the following useful property. 

\begin{fact}
\label{ancestor-fact}
 If two sets $A,~B\in\Pc$ with $A \leq B$ share a common edge $e$, then they also share the
 path from $e$ up to the highest edge in $B$, i.e., the edge of $B$ that is closest to the
 root.  
\end{fact}

Let $X_{j,S}$ and $Y_{j,S}$ for $j \in [n], S \in \cP$ to be two $\{0,1\}$-random
variables defined as follows: $X_{j,S}:= 1$ if and only if $S=S_j$ was assigned to
customer $j$ in step 1 of the procedure and and $Y_{j,S}=1$ if and only if $S$ survives in
Step 2, that is, $S_j=S$ and $j\in W$.  Note that, for $S\neq\emptyset$,
$\Pr[X_{j,S}=1]=\alpha x^*_{j,S}$.   

Now $\hx$ is defined by $\hx_{j,S_j}=1$ and $\hx_{j,S}=0$ for all $S\neq S_j$,
$j\in W$. 
We have  
\begin{align*}
\Pro{Y_{j,S} = 1} &= \Pro{X_{j,S} = 1} \cdot \Pro{Y_{j,S}  = 1 \mid X_{j,S} = 1} \\
		&= \alpha x^*_{j,S} \cdot \Pro{Y_{j,S}  = 1 \mid X_{j,S} = 1}\\
		& = \alpha x^*_{j,S} \cdot \bigl(1-\Pro{Y_{j,S}  = 0 \mid X_{j,S} = 1}\bigr).
\end{align*} 
We will show that $\Pro{Y_{j,S}  = 0 \mid X_{j,S} = 1}$, the probability of rejecting $S$
in step 2, is bounded from above by a constant. 

Recall that path $S$ is rejected in step 2 if its inclusion violates a capacity constraint at
some of its edges. It is natural now to apply a simple union bound on these
events. Unfortunately, this bound turns out to be too weak to prove a constant rejection
probability. Instead, we only consider a small subset $S' \subseteq S$ and show that the
rejection probability is bounded from above by the probability that some capacity
constraint on $S'$ is violated by the sets chosen in step 1. Let $v$ be the node in $S$
closest to the root.  
We consider the two branches of $S$ that are split by $v$, separately. Let $\ell =
(\ell_1, \ell_2,\dots)$ denote one branch and $r = (r_1, r_2, \dots)$ the other one, where
$\ell_1$ and $r_1$ denote the edges of $S$ incident to $v$. Note that $r$ or $\ell$ could
be empty if the path only consists of a single branch. The edges of $S'$ along the first
branch are now defined recursively: $\ell'_1 := \ell_1$ and $\ell'_i = \ell_j$ where $j =
\min\{k \mid c_{\ell_k} \leq c_{\ell'_{i-1}}/2\}$, i.e., $\ell'_i$ is the first edge after
$\ell'_{i-1}$ along the branch with less than half the capacity of
$\ell'_{i-1}$. Similarly, we define the edges along the second branch: $r'_1 := r_1$ and
$r'_i = r_j$ such that $j = \min\{k \mid c_{r_k} \leq c_{r'_{i-1}}/2\}$. 
So $S'=\bigl(\bigcup_i \ell'_i\bigr) \cup \bigl(\bigcup_i r'_i\bigr)$.  

A \textit{bad} event $\cE_e$ at edge $e$ occurs when $\sum_{j,A : e\in A} X_{j,A} \geq
c_{e}/2$. The next lemma shows that it is sufficient to only consider bad events at the
previously selected edges in $S'$.  

\begin{lemma}
\label{lem:upperbound_noninclusion}
 For every customer $j$, we have 
\begin{equation*}Pr[Y_{j,S} = 0 \mid X_{j,S} = 1] \leq \sum_{e \in S'} Pr[\cE_e].\end{equation*}
\end{lemma}
\begin{proof}
Assume that $S$ was rejected in step 2. In that case, there has to be an edge $e \in S$
such that $\sum_{j,A : e\in A} Y_{j,A} = c_e$, i.e., the number of paths picked prior to
$S$ that contain $e$ equals $c_e$ and therefore the inclusion of $S$ would violate the
capacity constraint on $e$. Let $e'$ be the next ancestor of $e$ that is in $S'$ (an edge
is also an ancestor of itself). Since the highest edge along each branch of $S$ was
included in $S'$, such an edge $e'$ has to exist. Moreover, $c_e \geq c_{e'}/2$: if $e \in
S'$, then $e=e'$; otherwise by definition, we have $c_e \geq c_{e'}/2$. Now by Fact
\ref{ancestor-fact}, every set $A$ that was considered in step 2 before $S$ also contains
$e'$.  Hence, we have 
\begin{align}
\sum_{j,A : e'\in A} X_{j,A} &\geq  \sum_{j,A : e'\in A} Y_{j,A}\geq \sum_{j,A : e\in A} Y_{j,A} = c_e \geq c_{e'}/2. \notag
\end{align}  

Thus the bad event $\cE_{e'}$ occurs at edge $e'$. The result then follows from a simple union bound.
\end{proof}

It remains to bound the probability of a bad event. 

\begin{lemma}
\label{lem:upperbound_badevents}
For $\alpha = 0.01$, we have $\sum_{e \in S'} \Pro{\cE_e} \leq 0.575$.
\end{lemma}

\medskip

\begin{proof}
Consider an edge $e$.
Define $Z_j=\sum_{A:e\in A} X_{j,A}$. Note that the random variables $Z_1,\ldots,Z_n$ are
independent. 
Then we have
$\Pro{\cE_e}= \Pro{\sum_{j} Z_{j} \geq c_e/2}$. 

%
%

By standard Chernoff bounds (see, e.g., \cite[Theorem 4.1]{MR95}), we get (with $\mu=\sum_j\EE[Z_j]\le\sum_j\sum_{A:e\in A} \EE[X_{j,A}]=\alpha\sum_j\sum_{A:e\in A} x^*_{j,A}\le \alpha c_e$) 
\begin{alignat}{1}
\Pro{\sum_{j} Z_j \geq c_e/2} &= \Pro{\sum_{j} Z_j \geq \bigl(1+\frac{c_e-2\mu}{2\mu}\bigr) \mu } \notag\\
&< \frac{(2\mu)^{c_e/2}e^{c_e/2-\mu}}{c_e^{c_e/2}}\notag\\
&< \left( (2\alpha)^{1/(2\alpha)}
e^\frac{1-2\alpha}{2\alpha}\right)^{\alpha c_e} \label{chernoff with parameters}\\ 
&< 0.231^{c_e} \qquad \qquad
\text{(since $\al=0.01$)} \notag 
\end{alignat} 
where inequality \eqref{chernoff with parameters} is valid for $\alpha < \frac{1}{2}$
since the left hand side is monotonically increasing in $\mu$ for such small
$\mu$. Finally, since the capacities of edges in $S'$ decrease by a factor of 2 along each
of the two branches, we get 
$\sum_{e \in S'} \Pro{\cE_e} \leq 2 \sum_{i \geq 0 } 0.231^{2^{i}} \leq 0.575$.
\end{proof}

Combining Lemma \ref{lem:upperbound_noninclusion} and Lemma
\ref{lem:upperbound_badevents}, we get $\Pro{Y_{j,S} = 1} \geq 0.00425 x^*_{j,S}$. Hence, the
expected weight of $\hx$ is at least $0.00425$ times the weight of $x^*$. 

\begin{remark}
 In the case of uniform capacities, Lemma \ref{lem:upperbound_noninclusion} simplifies to
 just two summands on the right hand side, that is, it is sufficient only to consider bad
 events on the first edge along each branch of $S$. Using this observation, the
 approximation factor can be improved to $0.3$. 
\end{remark}

\end{document}